\documentclass[reprint,superscriptaddress,aps,prl,amsmath,amssymb,nofootinbib]{revtex4-2}

\usepackage{xspace}
\usepackage{aas_macros}

\newcommand{\bs}{\boldsymbol}

\newcommand{\fstat}{FrankenStat\xspace}

\usepackage[explicit]{titlesec}
\titleformat{\section}[runin]{\normalfont\normalsize\itshape}{\thesection}{1em}{#1\textemdash}
\titlespacing{\section}{\parindent}{12pt}{0pt}

\titleformat{\subsection}[runin]{\normalfont\normalsize\itshape}{\thesubsection}{1em}{#1\textemdash}
\titlespacing{\subsection}{\parindent}{12pt}{0pt}

\usepackage{chngcntr}
\usepackage{natbib}
\usepackage{orcidlink}
\usepackage{xcolor}
\usepackage{appendix}
\usepackage{booktabs}
\usepackage{tikz}
\usetikzlibrary{shapes.geometric, arrows}
\tikzstyle{startstop} = [rectangle, rounded corners, 
minimum width=2.0cm, 
text width=1.5cm,
minimum height=1cm,
text centered, 
draw=black, 
fill=red!30]

\tikzstyle{io} = [trapezium, 
trapezium stretches=true, 
trapezium left angle=70, 
trapezium right angle=110, 
minimum width=3cm, 
minimum height=1cm, text centered, 
draw=black, fill=blue!30]

\tikzstyle{process} = [rectangle, 
minimum width=1cm, 
minimum height=1cm, 
text centered, 
text width=1.8cm, 
draw=black, 
fill=orange!30]

\tikzstyle{decision} = [diamond, 
minimum width=3cm, 
minimum height=1cm, 
text centered, 
draw=black, 
fill=green!30]
\tikzstyle{arrow} = [thick,->,>=stealth]

\usepackage{amsmath, physics2, amssymb, amsfonts, amsthm, siunitx, mathtools}
\DeclareSIUnit\arcsecond{as}

\usephysicsmodule{ab,diagmat,xmat}
\DeclareMathOperator{\Tr}{Tr}

\newtheorem{lemma}{Lemma}

\begin{document}

\title{\fstat I: a New Approach to Pulsar Timing Array Data Combination}

\author{David Wright\,\orcidlink{0000-0003-1562-4679}}
\affiliation{Department of Physics, Oregon State University, Corvallis, OR 97331, USA}

\author{Kalista Wayt\,\orcidlink{0000-0001-6630-5198}}
\affiliation{Department of Physics, Oregon State University, Corvallis, OR 97331, USA}

\author{Jeffrey S. Hazboun\,\orcidlink{0000-0003-2742-3321}}
\affiliation{Department of Physics, Oregon State University, Corvallis, OR 97331, USA}

\author{Xavier Siemens\,\orcidlink{0000-0002-7778-2990}}
\affiliation{Department of Physics, Oregon State University, Corvallis, OR 97331, USA}

\author{Rutger van Haasteren\,\orcidlink{0000-0002-6428-2620}}
\affiliation{Max-Planck-Institut f{\"u}r Gravitationsphysik (Albert-Einstein-Institut), Callinstra{\ss}e 38, D-30167 Hannover, Germany\\
Leibniz Universit{\"a}t Hannover, D-30167 Hannover, Germany}

\author{Levi Schult\,\orcidlink{0000-0001-6425-7807}}
\affiliation{Department of Physics and Astronomy, Vanderbilt University, 2301 Vanderbilt Place, Nashville, TN 37235, USA}

\author{Stephen R. Taylor\,\orcidlink{0000-0003-0264-1453}}
\affiliation{Department of Physics and Astronomy, Vanderbilt University, 2301 Vanderbilt Place, Nashville, TN 37235, USA}

\begin{abstract}
        In 2023, after more than two decades of searching, pulsar timing array (PTA) collaborations around the world announced evidence for a stochastic gravitational wave background. 
        It was quickly followed by work from the International Pulsar Timing Array (IPTA), demonstrating that the results of regional collaborations were consistent with each other. The combination of these datasets is still ongoing and represents a significant investment of time and expertise. 
        
        In that IPTA comparison, authors of this letter combined the separate datasets in the standard PTA optimal detection statistic for cross-correlations incoherently, that is, the data was combined without fitting a merged timing model across all PTA datasets, treating datasets of the same pulsar as independent, and neglecting the ``same pulsar, different datasets'' cross-correlations.
        This work refines that method by extending its core ideas beyond detection statistics and into a full, general data-combination method.
        We have demonstrated its efficacy and extreme efficiency on simulated data.
        This new method, \textit{FrankenStat}, is very similar in sensitivity and parameter-constraining power to traditional data combination methods while completing the full data combination in just a few minutes.
    
\end{abstract}

\maketitle

\section{\label{sec:intro}Introduction}
In recent years, a number of pulsar timing array (PTA) collaborations have reported increasing evidence for a stochastic gravitational wave background (GWB) at nanohertz frequencies \cite{ng15gwb,eptadr2_3:gwb,pptadr3:gwb,cptadr1_1:gwb,mpta4.5:gwb}.
This background of gravitational waves has many proposed origins, including cosmological and beyond Standard Model physics, \textit{e}.\textit{g}., inflation, cosmic strings, phase transitions, etc. \cite{ng15newphysics,Caprini:2018mtu,Domenech:2019quo,Kamionkowski:1993fg,Siemens:2006yp} or the unresolvable sum of millions of supermassive black hole binaries (SMBHBs) \cite{ng15smbbh,Burke-Spolaor:2018bvk,Kelley:2016gse,Ellis:2023dgf,EPTA:2023xxk}. While evidence continues to grow for the presence of a GWB, current PTA datasets still lack the sensitivity to differentiate between an astrophysical (SMBHBs) and cosmological origin. 

The various regional PTA collaborations together form the International PTA (IPTA) \cite{manchester:2013}, and the aggregated datasets are more sensitive to a GWB than the constituent datasets \cite{2016MNRAS.458.1267V,pdd+19}. Given the growing signal significance of the constituent PTA datasets, the consistency of the GWB parameters between the different sets of times of arrival (TOAs) that will comprise IPTA DR3 were studied in \cite{ipta3p+2024}. This mainly highlighted the promise of better spectral characterization. Additionally, a rudimentary version of the statistical method presented in this letter was used to show the potential of higher detection significance for a full DR3, giving a combined signal-to-noise ratio of $\rho\sim6.8$, as compared to the next largest individual-PTA signal-to-noise ratio of $\rho \sim 4.9$ \cite{ipta3p+2024}.

Historically, data from pulsars observed by more than one collaboration are combined by merging the TOAs and fitting a single deterministic timing model. Several intricacies are involved in combining these datasets that include: ensuring time standards and offsets are treated correctly between the various telescopes and receiver backends, assessing whether the more sensitive datasets necessitate additional parameters (\textit{e}.\textit{g}., relativistic binary evolution), and checking that noise models are consistent/appropriate for all pulsars \cite{2016MNRAS.458.1267V,pdd+19}. Data combination is time intensive and requires each pulsar to be handled by personnel with the necessary expertise. Owing to these complications, IPTA data releases often lag the regional PTA data releases by two to three years.

Recently, other strategies have been proposed to address the latency of full IPTA data combinations.
For example, \citet{ipta_lite} and \citet{ipta3p+2024} discuss the ``Lite'' method, which uses a Figure of Merit to select the most sensitive single-PTA dataset for each pulsar, enabling rapid preliminary analyses without full data combination.
While the Lite method selects the highest-quality subsets of data, \fstat enables the rapid mathematical combination of all available datasets, preserving the full sensitivity of the array without the computational overhead of a traditional unified timing model.

We use a new method to search for the GWB across pulsar datasets \emph{without} the need for the classic combination of TOAs.
The PTA analysis infrastructure is already set up to search for GWs across multiple pulsars, so with a few judicious choices, we can use much of the existing infrastructure to search for GWs across not only different pulsars but also different datasets from the same pulsar.
As will be shown with extensive simulations, this method retains comparable sensitivity and spectral properties of the traditionally combined data for GW searches---despite taking only a few seconds to combine a pulsar.
Combining this technique with the speedups enabled by new PTA analysis software, built on the backbone of GPU-accelerated and ML/AI-inspired techniques, means that the complete pipeline of data combination and downstream analyses can be done much faster than traditional methods allow (weeks instead of months). 

The remainder of the paper is laid out as follows. First we motivate the \fstat combination by revisiting how it was devised: examining the standard PTA likelihood for GWB searches and making modifications such that the data from \textit{all} PTAs are correctly modeled. We then introduce the \fstat combination method, validating it with simulation studies, and finally present results and implications for PTA science.

\section{The canonical PTA likelihood}
The PTA likelihood models TOAs as the sum of a deterministic pulsar timing model (TM) and multiple possible stochastic processes:
\begin{equation}
	\label{eq:toas}
	\bs{t} = \bs{t}_D + \bs{Fa} + \bs{n},
\end{equation}
where $\bs{t}_D$ is the deterministic timing delays, $\bs{Fa}$ is a Fourier series representation of the stochastic red noise processes (Fourier basis $\bs{F}$ and coefficients $\bs{a}$), and $\bs{n}$ contains white noise.

We then remove a best-fit timing model, $t_M$, from the TOAs to arrive at the \textit{timing residuals},
\begin{equation}
	\label{eq:stoch-part}
	\bs{\delta t} \equiv \bs{t}_D - \bs{t}_M + \bs{Fa} + \bs{n}= \bs{M\epsilon} + \bs{Fa} + \bs{n},
  \end{equation}
  where $\bs{M\epsilon}$ is a first-order, \textit{i}.\textit{e}., linear\footnote{In practice, many TM parameters are linear, \textit{e}.\textit{g}., the \textsc{JUMP} parameters of the pulsar. Hence, for the linear timing parameters, this perturbation model is exact.}, Taylor expansion of the TM, where $M_{ij}\equiv \left(\frac{\partial t_{M,i}}{\partial p_j}\right)_{\mathbf{p} = \mathbf{p}_0}$ is a matrix of the first derivatives of the TM wrt the paremeters evaluated at all TOAs and the best-fit parameters, $\mathbf{p}_0$, from a linear least squares analysis and $\bs{\epsilon}=\delta \bs{p}$ are the parameter deviations.
  This expansion accounts for deviations of the TM from the true values.
  Finally, assuming a multi-variate Gaussian with covariance ${\bs{N} = \ab< \bs{n} \bs{n}^T> }$ describes the distribution of $\bs{\delta t} - \bs{M\epsilon} - \bs{Fa}$, the log likelihood is
\begin{equation}
\begin{aligned}
	\label{eq:lnlikelihood}
	\ln{p{\ab(\bs{\delta t} \mid \bs{a}, \bs{\epsilon} )}} &= 
    -\frac{1}{2} \biggl [ \ln \vab { 2 \pi \bs{N} } \\
	& + \ab (  \bs{\delta t}\!-\!\bs{M\epsilon}\!-\!\bs{Fa} )^T\!\bs{N}^{-1} \ab ( \bs{\delta t}\!-\!\bs{M\epsilon}\!- \!\bs{Fa}  ) \biggr ] .
\end{aligned}
\end{equation}

In practice, we then perform a two step marginalization over the TM and Fourier coefficients \citep{NANOGrav:2023icp} that results in
\begin{equation}
\begin{aligned}
	\label{eq:lnlikelihood-marg}
	&\ln{p{\ab(\bs{\delta t} \mid \bs{\eta})}} = 
    -\frac{1}{2} \biggl [ \ln \vab { 2 \pi \bs{C} } + \bs{\delta t} ^T \bs{C}^{-1} \bs{\delta t} \biggr ],
\end{aligned}
\end{equation}
where $\bs{\eta}$ are Gaussian-process hyperparameters and 
\begin{equation}
\begin{aligned}
	\label{eq:lnlikelihood-marg-KD}
  \bs{C}^{-1} =& \bs{D}^{-1} - \bs{D}^{-1}\bs{F}\bs{\Theta}^{-1}\bs{F}^T\bs{D}^{-1}, \\
  \bs{D}^{-1} =& \bs{N}^{-1} - \bs{N}^{-1}\bs{M}\bs{\Delta}^{-1}\bs{M}^T \bs{N}^{-1}, \\
  \bs{\Theta} \equiv&  \ab <\bs{a}\bs{a}^T>^{-1} + \bs{F}^T \bs{D}^{-1}\bs{F},\\
  \bs{\Delta} \equiv& \ab <\bs{\epsilon}\,\bs{\epsilon}^T>^{-1}\, + \bs{M}^T\bs{N}^{-1}\bs{M}.\\
\end{aligned}
\end{equation}

If we do not vary the white noise, the inverse $\bs{D}^{-1}$ can be cached in its entirety. The computational burden of the likelihood is then shifted to $\bs{\Theta}^{-1}$, which is considerably easier to calculate \citep{NANOGrav:2023icp}. The covariance matrix of the Fourier coefficients, $\langle\bs{a}\bs{a}^T\rangle$, represents the Fourier Expansion of a stationary Gaussian Process with power spectral density of the RN or GWB---usually a power law. The covariance matrix for the timing model parameter perturbations, $\langle\bs{\epsilon}\,\bs{\epsilon}^T\rangle$, is also treated as diagonal, \textit{i}.\textit{e}., with no correlation between perturbations, and is ``analytically marginalized'' by setting the diagonal variances $\langle\epsilon_i\,\epsilon_i^T\rangle=\infty$ \cite{2009MNRAS.395.1005V,2016ApJ...821...13A}. This then implies $\langle\epsilon_i\,\epsilon_i^T\rangle^{-1}=0$, and the expressions above simplify. This marginalization is a powerful tool that gives immense freedom of variations to the TM parameters, allowing searches for GWs without hindrance from an ill-fit TM.

\section{\fstat data combination}
Typically, PTA data combination involves combining TOAs and fitting a single, merged TM.
The \fstat method of combining data, unlike the typical strategy, does not require a merged TM.
We instead opt to concatenate the residuals,
\begin{equation}
	\bs{\delta t} =
    \bxmat[showtop=2,format=#1_{#2}]{\bs{\delta t}}{n}{1},
\end{equation}
assuming that each PTA (denoted as subscripts) has adequately fit a TM to each pulsar.
After concatenating the residuals, we can proceed to build the likelihood.

For now, we limit the residuals to originate from one pulsar in possibly many different PTAs.
To correctly model our new, concatenated residuals we must also combine the TM perturbations $\bs{M \epsilon}$, Fourier series $\bs{Fa}$, and white noise covariance $\bs{N}$.

The Fourier design matrix $\bs{F}$ is~${N_{TOA} \times 2N_{freq}}$, consisting of alternating sines and cosines evaluated at each TOA and modeled frequency where the frequencies are based on the total timing baseline of the pulsar.
The correct combination of $\bs{F}$ is therefore just the concatenation along the first dimension,
\begin{equation}
	\bs{F} =
    \bxmat[showtop=2,format=#1_{#2}]{\bs{F}}{n}{1},
\end{equation}
and the coefficients $\bs{a}$ are unaffected by this scheme.
Note that this concatenation, along with the following white noise covariance matrix combination, are exactly the same for a typical data combination as well.

The white noise covariance $\bs{N}$ is modeled per receiver/back-end combination \citep{NANOGrav:2023icp,NANOGrav:2023ctt}.
To construct a combined $\bs{N}$, we build it as usual while including blocks for each receiver/back-end combination in the new concatenated set of residuals.
This will result in a block diagonal, square matrix:
\begin{equation}
	\label{eq:franken-N}
	\bs{N} = \bdiagmat[empty={}]{\bs{N}_1, \bs{N}_2, \ddots, \bs{N}_n}.
\end{equation}

The biggest difference in this new analysis is that the TM Taylor expansion term, $\bs{M \epsilon}$, requires a different combination method.
The TM design matrix $\bs{M}$ is~$N_{TOA} \times N_{par}$, consisting of the first derivative of each TM component evaluated at the least-squares fit and each TOA for a given pulsar.
The second dimension of $\bs{M}$ will typically not be the same for the same pulsar in different PTAs.
Furthermore, even if the dimension agreed, the least-squares fit may be different for each PTA and result in different first derivatives.
Thus, it is incorrect to simply concatenate the design matrices $\bs{M}$ along the TOA dimension.
We will instead construct a \textit{block diagonal} design matrix:
\begin{equation}
	\label{eq:franken-M}
	\bs{M \epsilon} =
    \bdiagmat[empty={}]{\bs{M}_1, \bs{M}_2, \ddots, \bs{M}_n}
    \bxmat[showtop=2,format=#1_{#2}]{\bs{\epsilon}}{n}{1},
\end{equation}
This is the simplest choice for treating the TM perturbations. It is also a conservative one. By ignoring that some of these parameters are the same physical process for a given pulsar, or that others are correlated, we add freedom/uncertainty to the model that can only lead to a loss of information. As we will show, this has \emph{minimal} effect on the outcome of the analysis. 

Note that some of these techniques were previously applied to single pulsar analyses of IPTA DR2. Known as the Borg method, it stitched constituent pulsars together with common intrinsic red noise, dispersion measure, and GWB parameters \cite{Schult2021}. \fstat builds on this foundation, delivering a comprehensive approach to incoherent data combination.

These steps can be repeated per-pulsar, per-PTA.
The resulting concatenated and block-diagonal quantities can be used in all typical PTA likelihoods.  

\subsection{Implementation}
The most efficient way to implement this combination such that it appears correctly in \eqref{eq:lnlikelihood} is to build combined \textit{pulsar objects}.
Our analysis codes \citep{enterprise,discovery}, create pulsar objects (in the sense of object oriented programming) that contain noise parameters, design matrices, TOAs, and residuals.
By including our new concatenated and block-diagonal quantities in these objects, the likelihood is \textit{automatically} constructed correctly.
This means that any downstream analyses function exactly the same with a \fstat-combined pulsar object (FrankenPulsar\footnote{The ``stitching together'' of these pulsar datasets leads to our shorthand name. We leave it to the reader to decide whether the adjective refers to the authors or the pulsars.}.) as they would with a standard pulsar object, including all downstream Bayesian analyses.
We stress that no code modifications, beyond the pulsar objects themselves, are necessary to use the \fstat combination with \emph{any} PTA analysis code that uses standardized pulsar objects, such as \textsc{enterprise} and {\textsc{enterprise-extensions}} \citep{enterprise,enterprise-extensions}.

To be explicit, we only need to perform three operations to create a FrankenPulsar.
Start with a collection of un-combined pulsar objects.
Then,
\begin{samepage}
\begin{enumerate}
  \item Iterate over the objects' attributes and concatenate any that are vectors of length $N_{TOA}$.
  \item Construct the block diagonal TM design matrix $\bs{M}$.
  \item Truncate the sky position attribute at the precision where all inputs agree. This is necessary for the calculation of spatial correlations.
\end{enumerate}
\end{samepage}
Operations like the concatenation of the Fourier design matrices and the creation of block-diagonal white noise matrices happen automatically in the analysis codes, as they would in a full combination.

\section{Simulation study}
\label{sec:simulations}
To validate the \fstat data combination, we simulate one hundred traditionally-combined and \fstat-combined PTA datasets.
A high-level overview of our pipeline is given in Figure \ref{fig:pipeline}.

\begin{figure}[t]
	\begin{tikzpicture}[node distance=1.5cm]

		\node (simcomb) [startstop] {Simulate PTA};
		\node (spnacomb) [process, below of=simcomb] {SPNA};
		\node (refitcomb) [process, below of=spnacomb] {Refit TMs with Noise};
		\node (split) [startstop, below of=refitcomb, text width=2.0cm] {Split PTAs \& Refit TMs};

		\node (spna-ptas) [process, right of=spnacomb, xshift=1.3cm] {SPNA PTAs};
		\node (refit-ptas) [process, below of=spna-ptas] {Refit TMs with Noise};

		\node (frankenize) [startstop, below of=refit-ptas] {Frankenize \& SPNA};

		\node (maxlike) [process, right of=refit-ptas, xshift=1.3cm, yshift=0.18cm] {Maximum
			a Posteriori
			all PTAs};

		\node (detstat) [startstop, right of=frankenize, xshift=1.3cm, fill=blue!30] {Detection
			Statistics};

		\draw [arrow] (simcomb) to (spnacomb);
		\draw [arrow] (spnacomb) -- (refitcomb);
		\draw [arrow] (refitcomb) -- (split);
		\draw [arrow] (split.north east) to (spna-ptas.west);
		\draw [arrow] (spna-ptas) to (refit-ptas);
		\draw [arrow] (refit-ptas) to (frankenize);
		\draw [arrow] (frankenize.north east) to (maxlike.west);
		\draw [arrow] (maxlike) to (detstat);

	\end{tikzpicture}
	\caption{
		Flow chart detailing pipeline.
		First, a baseline, fully ``Combined'' PTA is simulated.
		Then, a single pulsar noise analysis (SPNA) is ran on each pulsar using gradient descent to find maximum \textit{a posteriori} (MAP) noise parameters, and the timing model (TM) is refit.
		To construct the split PTAs that we will later recombine with \fstat, we take every third TOA from the original data to create three new PTA datasets.
		We then run an SPNA and TM refit on every pulsar in these new, split datasets.
		Once the TMs are refit, we can ``Frankenize'' the pulsars to create FrankenPulsars.
        The FrankenPulsars should give us noise estimates that are more accurate than the individual PTAs and comparable to the Combined pulsars.
        To benefit from this increased sensitivity, we run SPNAs on the \fstat pulsars as well.
	    In order to calculate detection statistics, we also need estimates of the GWB.
		We find a MAP estimate using a CURN likelihood for the Combined PTA, each split PTA, and the \fstat PTA.
		These MAP parameters are then used to find the SNR and corresponding $p$-value from the null distribution.
	}
	\label{fig:pipeline}
\end{figure}
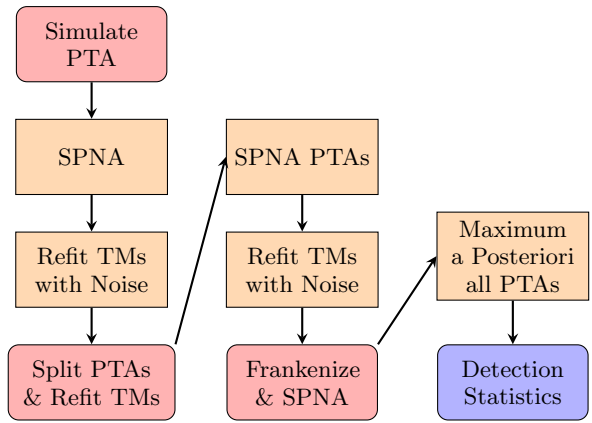

We begin by simulating a fully ``Combined'' PTA with 126 pulsars uniformly distributed on the sky and a twenty year baseline with a cadence of three observations per month.
This ``Combined'' PTA is our proxy for a traditionally-combined PTA dataset.
We build a TM by hand that includes sky position, parallax, and quadratic spindown.
Each of these parameters, other than sky position, is drawn from a uniform distribution given in Table \ref{tab:param-dist}. Uniform sky positions were created via a Fibonacci lattice. For a given number of pulsars, this is always the same lattice. Therefore, our pulsar positions are constant across simulations.
After building the TM, we load our simulated pulsars into \textsc{pta\_replicator} \citep{ptareplicator}.
Then, using \textsc{pta\_replicator}, we inject intrinsic red noise, a gravitational wave background, and white noise. 
The parameters for these signals are also drawn from distributions given in Table \ref{tab:param-dist}.
For each simulation, we draw new TM, signal, and noise parameters.
The GWB parameters are held fixed for all of our simulations.
We choose a spectral index of $\gamma = 13 / 3$ to reflect the expected value for a population of SMBHBs evolving purely due to GW emission and an amplitude of $\log_{10} A = -15.7$ to give $\geq5\sigma$ GWB detections in our Combined simulations.

\begin{table}[t]
  \centering
\begin{tabular}{lllc}
  \toprule
  Parameter & Distribution & Fractional Unc. & Units \\
  \midrule
  $F$ & Uniform(2.47, 3) & Uniform(-15, -13) & $\log_{10}$~\unit{\hertz} \\
  $\dot{F}$ & Uniform(-17, -15) & Uniform(-5, -3) & $\log_{10}$~\unit{\hertz\squared} \\
  Parallax & Uniform(0.5, 2) & Uniform(0.1, 0.8) & \unit{\milli\arcsecond} \\
  \midrule
  GWB $\gamma$ & Fixed: 13/3 & & \\
  GWB $\log_{10} A$ & Fixed: -15.7 & &\\
  RN $\gamma$ & Uniform(1, 6) & &\\
  RN $\log_{10} A$ & Uniform(-20, -12) & &\\
  TOA Error & Normal(100, 10)& & \unit{\nano\second}\\
  \bottomrule
\end{tabular}
\caption{
  Table of parameters used to build simulated PTA datasets.
  For each parameter, we draw a value and a fractional uncertainty.
  The uncertainty on the parameter will be equal to the value times the fractional uncertainty in linear space.
  For our log space parameters, we add the fractional uncertainty.
  For each simulation, the uniform sky positions are created via a Fibonacci lattice.
  For a given number of pulsars, this is always the same lattice.
  Therefore, our pulsar positions are constant across simulations.}
\label{tab:param-dist}
\end{table}
In order to realistically fit the TM, we run a single pulsar noise analysis (SPNA) and refit the TM using the maximum \textit{a posteriori} MAP noise parameters for each pulsar.
A traditional SPNA requires an expensive Markov chain Monte Carlo (MCMC) run in order to find estimates of the noise parameters.
We need to run an SPNA for each pulsar in each simulation, which becomes a major computational bottleneck in our simulation study.
Instead of MCMC, we opt to utilize the automatic differentiation present in \textsc{discovery} to find (MAP) noise parameters via gradient descent.
In combination with \textsc{discovery} \citep{discovery}, we use \textsc{Optax} \citep{deepmind2020jax} and \textsc{NumPyro} \citep{phan2019composable,bingham2019pyro} to build our optimization routines.
In our testing, this method accurately returns the modes of the MCMC posteriors while having many orders of magnitude shorter runtime\footnote{Note that a few of the gradients in these fits are comparatively large. Particularly, we found some WN parameters, e.g. EFAC, get stuck in local minima at the edges of the prior range. These few problematic parameters greatly skewed the detection statistics when left unfixed.\label{efac_footnote}}.

After acquiring noise estimates and refitting TMs, we split our Combined PTA into three separate PTAs.
There are many possible ways to split the data---here we choose to take the TOAs and assign them to PTA 1, PTA 2, and PTA 3 in order repeatedly until we exhaust the number of TOAs.
This results in 126 pulsars in each PTA with $1/3$ of the total number of TOAs.
The split PTAs then require TM refits followed by SPNAs and another round of TM refits.

With the fully split and refit PTAs ready, we can combine them in the \fstat fashion.
As previously described, this results in combined pulsar objects that we can then pass on to any of our PTA analysis codes.
However, the \fstat-combined pulsars should also have stronger constraints on noise parameters than the split PTAs.
We therefore run SPNAs on each \fstat pulsar and update their stored noise parameters.

Now that our simulations are complete, we wish to compute detection statistics.
We choose to use the PTA optimal statistic to compute SNR and $p$-values, as a standard in the field \citep{Chamberlin:2014ria, noise_marg} with which to compare.

These calculations also require parameter estimates, much like the TM refits.
However, the parameter estimates that we now need are not per-pulsar but must be calculated across the entire PTA.
Again, in order to find these estimates at scale, we find MAP estimates using \textsc{discovery}.
For computational purposes, we use a likelihood with a common uncorrelated process (CURN) instead of Hellings-Downs (HD) correlations.

Given the MAP estimates, we can compute the SNR and detection significance ({$p$-value}) for each of our simulated PTAs.
These calculations are done in \textsc{discovery}, using the rank-reduced methods introduced by \citet{vanHaasteren:2025jvo}.
The rank-reduced methods enable analytic {$p$-value} calculations that were previously computationally prohibitive.

\begin{figure}[t]
	\centering
	\includegraphics[width = \columnwidth]{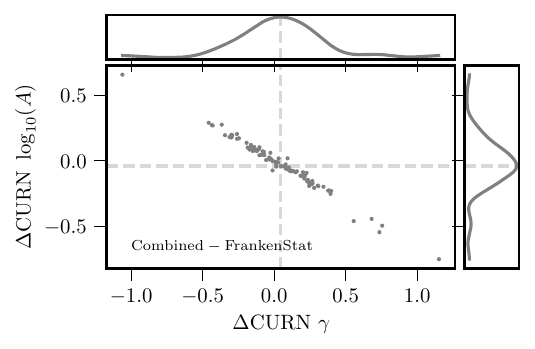}
	\caption{
		Difference between Combined and \fstat CURN maximum likelihood estimates for one hundred simulations.
		The estimates are in close agreement, with $\mu \pm \sigma$ of $-0.03 \pm 0.19$ and $0.05 \pm 0.28$ for $\log_{10}\ab(A)$ and $\gamma$, respectively.
	}
	\label{fig:curn-max-like}
\end{figure}

\begin{figure}[b]
	\centering
	\includegraphics[width = \columnwidth]{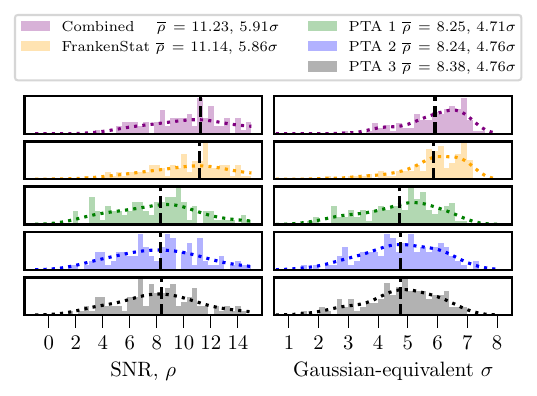}
	\caption{
		Distribution of SNRs and $p$-values corresponding to the maximum likelihood parameters in Figure~\ref{fig:curn-max-like}.
        Dotted lines represent 1D kernel density estimators, and black dashed lines are the means of the distributions.
		The \fstat distribution show excellent agreement with the Combined data.
		Both \fstat and the traditionally-combined data show a marked improvement in SNR and detection significance over the individual datasets.
	}
	\label{fig:snrs-pvals}
\end{figure}

\begin{figure}[ht]
	\centering
	\includegraphics[width = \columnwidth]{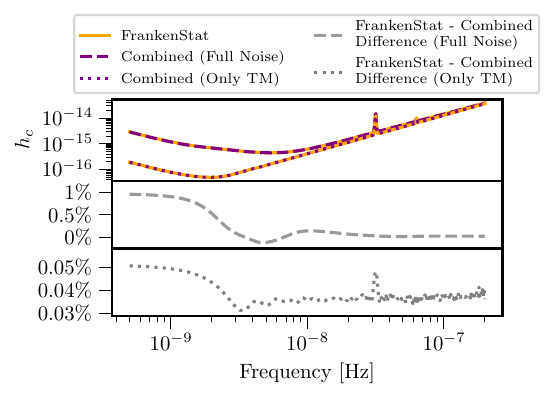}
	\caption{
		PTA sensitivity curves from a representative simulation of the Combined and \fstat data.
        In this plot, we show two types of sensitivity curve: one including the effects of the timing model fit, intrinsic red noise, white noise (WN), and GWB self-noise (``Full Noise''), and one with only timing model effects and WN (``Only TM'').
        For the sensitivity curve with TM + WN, the differences between the Combined and \fstat curves are due only to the different TMs, since the WN is kept constant. 
		These sensitivity curves are nearly identical and show at most a 0.05\% difference.
		In this simulation, the difference is always positive, which means that \fstat is slightly less sensitive but negligibly so when only accounting for the TM.
        When including the noise, the differences become only slightly larger and peak at about $1\%$.
        These differences are now mostly due to the separate noise parameters that come from the Combined and \fstat MAP noise and MAP CURN estimates.
	}
	\label{fig:sensitivity}
\end{figure}

\section{Results}
\label{sec:results}
Before analyzing the detection statistics, we check that the \fstat and Combined PTAs recover the same estimated GWB parameters.
As seen in Figure \ref{fig:curn-max-like}, the \fstat combination shows excellent agreement with the Combined data.

We now calculate the distribution of SNR for an HD correlated GWB over the one hundred simulations of each individual (PTA 1--3), Combined, and \fstat PTA we have simulated.
The histograms and 1D kernel density estimators for these distributions are shown in Figure \ref{fig:snrs-pvals}.
As expected, the individual, split PTAs have similar but smaller SNRs than the Combined data.
The Combined and \fstat SNRs differ by less than 1\%, showing that the \fstat combination has effectively the same sensitivity to the GWB. This is clearly seen in the sensitivity curves in Figure \ref{fig:sensitivity}.
The \fstat combination has minimal impact on the dataset's sensitivity to a GWB.
In certain cases, the difference in sensitivity can be proved analytically.
One such case is discussed in the Appendix. 

The SNR alone does not tell us about the significance of the measurement.
We calculate the probability that the SNR could have been measured under the null hypothesis: the $p$-value.
For each $p$-value, we find the Gaussian-equivalent $\sigma$ and plot the distribution in Figure \ref{fig:snrs-pvals}.
Again, as with the SNR, the distribution of $p$-values is almost one-to-one: the Combined and \fstat data differ by less than 1\%.
Compared to the individual PTAs, the Combined and \fstat data gain over $1\sigma$ of significance.

\section{Conclusions}
\label{sec:conclusions}
We have developed a PTA data combination method, \textit{\fstat}, that is simple to implement and have demonstrated its efficacy.
Contrasting usual PTA data combination strategies that take years and considerable computational resources, \fstat pulsar combinations can be run on virtually any modern computer in minutes. 

\fstat is not only fast and efficient---it also has similar sensitivity to gravitational waves as a traditionally combined dataset.
The implementation of \fstat used in this letter is available on GitHub\footnote{https://github.com/davecwright3/frankenstat-paper-1}.
For readers interested in using \fstat in PTA analyses, we recommend using the Python package \textsc{MetaPulsar} \citep{metapulsar}.
\textsc{MetaPulsar} contains an implementation of \fstat, alongside other data-combination utilities and strategies.
Going forward, the reference \fstat implementation will be maintained there. 

Late in the preparation of this work, an analysis of IPTA DR3 datasets was posted \cite{Yu:2025dor} which uses \textsc{MetaPulsar} in a different mode that has yet to be validated in publicly available work. One clear advantage of \fstat over the method used in \cite{Yu:2025dor} is that it sidesteps any questions about which dataset's timing model parameters will be used. We also note that the anomalous WN parameters reported therein could be related to the challenges we point out with MAP-estimate convergence (see fn. \ref{efac_footnote}).

As an extension to the \fstat work shown here, we can treat datasets of the same pulsar as independent and neglect the ``same pulsar, different datasets'' cross-correlations in the calculation of the standard PTA optimal detection statistic. This approach to data combination, which is what was originally applied in the IPTA comparison paper \cite{ipta3p+2024}, is what we have named, ``The Covariance Approach.'' A detailed derivation that highlights corrections to the statistic's variance that were neglected in the original IPTA comparison will be described here \cite{Frankstatcov}.

We encourage the adoption of our method in PTA data analysis pipelines. Once constructed, FrankenPulsars can be used for any signal searches across multiple datasets. 

\section{Acknowledgments}
The authors are funded as part of the NANOGrav Collaboration through the National Science Foundation (NSF) NANOGrav Physics Frontiers Center award \#2020265. J.S.H. and K.W. acknowledge support from NSF CAREER award \#2339728. J.S.H, K.W and D.W. acknowledge support from an Oregon State University startup fund. S.R.T and L.S acknowledges support from NSF AST-2307719. S.R.T acknowledges support from an NSF CAREER award \#2146016, and a Chancellor's Faculty Fellowship from Vanderbilt University.

\bibliography{main}

\appendix
\counterwithin{equation}{section} 
\setcounter{equation}{0}
\counterwithin{equation}{subsection}
\renewcommand{\theequation}{A\arabic{equation}}
\section{Appendix: Proof of \fstat SNR upper bound for nested case}
\label{appendix:proof}
Here we prove that if the blocks of a block-diagonal \fstat TM design matrix comprise \textit{exactly} the Combined, non-block diagonal TM design matrix, then the \fstat SNR is \textit{less than or equal to} the Combined SNR.
In other words, this proof is valid if both the FrankenStat and Combined TM design matrix are produced from the \textit{same TM fit}. This would be the case if the individual PTAs' TM fits are the same as the Combined TM fit.

The proof proceeds by showing that, in the case described, the inverse-noise-weighted transmission function \( \mathcal{N}_I^{-1}(f) \) \citep{Hazboun:2019vhv} for the \fstat combination is always less than or equal to the full data combination for any given pulsar \( I \)\,.
This quantity appears in the expression for the expected SNR denoted as \(\rho\):
\begin{equation}
	\label{eq:ncal-snr}
	\rho^2 \simeq \sum_I \sum_{J>I}T_{IJ}\chi^2_{IJ} \int_0^{f_\textrm{Nyq}} \textrm{d}f \frac{S_h^2(f) \mathcal{R}^2(f)}{\mathcal{N}_I(f)\mathcal{N}_J(f)},
\end{equation}
where \( S_h(f) \equiv P_h(f) / \mathcal{R}(f) \), \(P_h = \frac{A^2_\textrm{gw}}{12\pi^2} \left(\frac{f}{f_\textrm{yr}}\right)^{2\alpha}f^{-3}\) is the one-side strain PSD, and \( \mathcal{R}(f) =  \left(12\pi^2f^2\right)^{-1}\) is the timing-residual response function to a monochromatic plane GW averaged over inclination, sky position, and polarization.

\citet{Hazboun:2019vhv} define \( \mathcal{N}_I^{-1}(f) \) as
\begin{equation}
	\label{eq:ncal-def}
	\mathcal{N}_I^{-1}(f) = \frac{1}{2T_I}\Tr\left[ F_I^\dagger G_I\left(G_I^\dagger C_I G_I \right)^{-1}G_I^\dag F_I\right],
\end{equation}
where \( G_I \) is an orthonormal basis for the null space of the pulsar's TM, \( C_I \) is the time-domain covariance matrix for the pulsar, and \( F_I \) is a column vector of complex exponentials evaluated at each TOA for a given frequency \( \left(F_I\right)_k \equiv \exp\left(i 2\pi f t_k\right)\).

The quantity in the trace is very nearly a projection operator, but not quite.
There is a nice property involving the trace of a projection operator.

\begin{lemma} \label{lemma:psd-trace}
	The trace of a projection matrix times a positive semidefinite matrix is greater than or equal to zero.
\end{lemma}
\begin{proof}
	A projection matrix is positive semidefinite.
	It suffices to prove that the trace of the product of two positive semidefinite matrices is non-negative.
	Let \( A,X \) be two positive semidefinite matrices.
	We can take a unique matrix square root of either.
	Proceed with the square root of \( A \).
	The trace becomes
	\begin{align}
		\Tr \left[ AX \right ] & =\Tr \left[ \sqrt{A} \sqrt{A} X \right ]                                   \\
		                       & =\Tr\left[\sqrt{A}X\sqrt{A}\right]                                         \\
		                       & = \delta_{il}\sqrt{A}_{ij} X_{jk}\sqrt{A}_{kl}                             \\
		                       & = \sqrt{A}_{ij} X_{jk}\sqrt{A}_{ki}                                        \\
		                       & = \sum_i \sqrt{A}_{\textrm{row}\, i}\, X\, \sqrt{A}_{\textrm{col}\, i}\, ,
	\end{align}
	where each term in the sum must be non-negative by the definition of a positive semidefinite matrix.
	Thus, the trace of a projection matrix times a positive semidefinite matrix is greater than or equal to zero.
\end{proof}

Now, we need to reform the argument of the trace to a projection operator times a positive semidefinite matrix to use this property.
We can accomplish this with a \textit{whitening} transformation.
From here on, we will drop the pulsar index for brevity.
Define
\begin{equation}
	\label{eq:whitening-def}
	H \equiv C ^{1/2}G.
\end{equation}
\( C ^{1/2} \) and \( C ^{-1/2} \) exist because \( C \) is positive definite.

Inserting this quantity into the trace in Equation \eqref{eq:ncal-def} results in
\begin{equation}
	\Tr\left[ F^\dagger C ^{-1/2}H\left(H^\dagger H \right)^{-1}HC ^{-1/2}F\right].
\end{equation}
Then, a cyclic permutation results in
\begin{equation}
	\Tr\left[H\left(H^\dagger H \right)^{-1}HC ^{-1/2}F F^\dagger C ^{-1/2}\right].
\end{equation}
For clarity, we define \( A \equiv  C ^{-1/2}F F^\dagger C ^{-1/2}\).
This matrix is positive semidefinite.
The proof of this follows easily from the observation that \( F F^\dagger \) is positive semidefinite with a zero eigenvalue of multiplicity \( N_{\textrm{TOA}} -1 \) and a single non-zero eigenvalue equal to \( N_{\textrm{TOA}} \).
Our trace is now
\begin{equation}
	\Tr\left[H\left(H^\dagger H \right)^{-1}HA\right].
\end{equation}

We now consider the \( H \) defined by the \fstat combination, call it \( H' \).
This differs from the base-case by the construction of the TM design matrix.
The base-case design matrix is dense, but the \fstat design matrix is block-diagonal.
The column space of the base-case is a sub-space of the \fstat case.
Thus, the null space of the \fstat case is a sub-space of the base-case: \( \textrm{col}(G') \subset \textrm{col}(G)\)\,.
\( H \) is isomorphic to \( G \), so the same inequality holds for \( H \) and \( H' \)\,.

We wish to find an inequality between \( \mathcal{N}^{-1}(f) \)  and \( \mathcal{N}^{\,'-1}(f) \).
Define two projection operators
\begin{equation}
	P(H) = H \left(H^\dag H\right)^{-1} H^\dag
\end{equation}
and
\begin{equation}
	\widetilde{P}(H) = P(H) - P(H').
\end{equation}
It is idempotent:
\begin{align}
	\widetilde{P}(H)\widetilde{P}(H)  = & \left(P(H) - P(H')\right) \left(P(H) - P(H') \right) \\
	=                                   & P(H)P(H) - P(H')P(H)                                 \\
	\nonumber                           & - P(H)P(H') + P(H')P(H')                             \\
	=                                   & P(H) - P(H') - P(H') + P(H')                         \\
	=                                   & P(H) - P(H')                                         \\
	=                                   & \widetilde{P}(H),
\end{align}
where we have used the fact that \( \textrm{col}(H') \subset \textrm{col}(H)\) to reduce \( P(H')P(H)\) and \( P(H)P(H') \) to \( P(H') \)\,.

The hermiticity is trivial: \(  \widetilde{P}(H) = P(H) - P(H') = P^\dag(H) - P^\dag(H') =  \widetilde{P}^\dag(H)\).

We can now relate the traces as follows:
\begin{align}
	\Tr\left[P(H)A\right] & = \Tr\left[\left(\widetilde{P}(H) + P(H')\right)A\right]    \\
	                      & = \Tr\left[\widetilde{P}(H)A\right] + \Tr\left[P(H')A\right] \\
	                      & \geq \Tr\left[P(H')A\right].
\end{align}
The inequality follows from Lemma \ref{lemma:psd-trace} because the trace of \( \widetilde{P}(H)A \) must be non-negative.
Therefore,
\begin{equation}
	\mathcal{N}^{-1}(f) \geq \mathcal{N}^{\,'-1}(f),
\end{equation}
and from the definition of the SNR in Equation \eqref{eq:ncal-snr},
\begin{equation}
	\rho \geq \rho'.
\end{equation}
\qed

\end{document}